%% file: root.tex
\documentclass[journal]{IEEEtran}
\usepackage{mathrsfs}
\usepackage{amsthm} 
\usepackage{amsmath,amsfonts}
\usepackage{array}
\usepackage{url}
\usepackage{graphicx}
\usepackage[caption=false,font=footnotesize]{subfig}
\usepackage{algorithm}

\usepackage{lineno,hyperref}
\usepackage{moreverb}
\usepackage{cite}

\usepackage{epstopdf}
\usepackage{xcolor}
\usepackage{indentfirst} 
\usepackage{float}
\usepackage{colortbl,booktabs}

\usepackage{multirow}
\usepackage{amssymb}
\usepackage{soul}

\usepackage{tabularx}
\usepackage{makecell}

\usepackage{tikz}
\usepackage{pgfplots}
\pgfplotsset{compat=newest}

\def\BibTeX{{\rm B\kern-.05em{\sc i\kern-.025em b}\kern-.08em
    T\kern-.1667em\lower.7ex\hbox{E}\kern-.125emX}}
\usepackage{balance}
\usepackage{enumitem}

\usepackage{newtxtext}

\newtheorem{theorem}{Theorem}
\newtheorem{lemma}{Lemma}
\newtheorem{definition}{Definition}
\newtheorem{remark}{Remark}
\newtheorem{assumption}{Assumption}
\newtheorem{proposition}{Proposition}

\newtheorem{problem}{Problem}

\usepackage{hyperref} \hypersetup{pdfstartview=FitH,
	CJKbookmarks=true,
	bookmarksnumbered=true,
	bookmarksopen=true,
	colorlinks,
	pdfborder={0 0 1},
	linkcolor=blue,
	anchorcolor=blue,
	citecolor=blue,} 

\begin{document}
\title{Data-Driven Synthesis of Robust Positively Invariant Sets:\\ From State Feedback to Output Feedback}
\author{
		
		Zhijie~Ning


		\thanks{Zhijie~Ning is with the Department of Control Science and Engineering, Harbin Institute of Technology, Harbin 150001, China (e-mail: ningzj@stu.hit.edu.cn).}
		}

\maketitle

\bibliographystyle{IEEEtran}

\begin{abstract}
	This paper develops a direct data-driven framework for robust positively invariant (RPI) set synthesis for unknown linear time-invariant systems under state-feedback and observer-based output-feedback scenarios.
	The feedback and observer gains, along with the RPI sets, are directly synthesized from noisy offline data by solving semidefinite programs (SDPs), avoiding intermediate model identification or explicit model-uncertainty set construction.
	In the state-feedback case, a linear-quadratic (LQ)-type feedback gain is first computed, and an ellipsoidal RPI set is then synthesized for the resulting closed-loop dynamics.
	In the observer-based output-feedback case, offline data are used to compute the observer gain and the corresponding RPI set for the system state through an augmented-state formulation.
	This design provides a unified method for invariant-set computation in both scenarios, and the direct data-driven formulation avoids the explicit construction and propagation of an intermediate model-uncertainty set.
	Numerical examples illustrate the effectiveness of the proposed method.
\end{abstract}

\input{body/1.Introduction}

\input{body/2.Preliminary}
\input{body/3.IS}
\input{body/4.IO}

\input{body/5.Sim}

\section{Conclusion}\label{sec:Conclusion}
This paper presented a direct data-driven framework for the synthesis of RPI sets for unknown linear time-invariant systems.
The decoupled synthesis architecture and the S-procedure avoided the non-convex coupling between feedback gains and invariant sets, and the resulting synthesis conditions were formulated as tractable convex SDPs after fixing scalar multipliers.
This approach retained the nominal LQ-type gain-design objective while circumventing the reliance on explicit model identification and model uncertainty set construction, thereby avoiding the explicit construction and propagation of an intermediate model-uncertainty set.
Furthermore, the framework was extended to the practical observer-based output-feedback scenario, providing invariant bounds for system states with limited online state observability.

Future research will explore the systematic tuning of hyperparameters to further reduce geometric conservatism and will incorporate the proposed invariant-set construction into data-driven predictive control schemes to address closed-loop constraint satisfaction.

\bibliography{reference}

\end{document}

%% file: body/1.Introduction.tex
\section{Introduction}
The theory of invariant sets, particularly robust positively invariant (RPI) sets, provides a fundamental tool for ensuring worst-case safety of disturbed control systems \cite{rakovic2007optimized, blanchini1999set}. 
By ensuring that system trajectories remain within prescribed sets despite admissible disturbances, RPI sets play a central role in tube-based model predictive control (MPC) through constraint tightening \cite{mayne2005robust} and in trajectory tracking through robust error bounds \cite{chen2021data}.
Therefore, computing less conservative RPI sets is essential for reducing the conservatism of robust control designs.

Historically, computing tight RPI approximations fundamentally depended on exact analytical models, employing techniques such as Minkowski sum methods \cite{rakovic2005invariant}, linear programming \cite{trodden2016one}, and Lyapunov-based techniques \cite{nazin2007rejection}. 
However, obtaining accurate first-principles dynamics for real-world systems is often infeasible, motivating the shift toward data-driven synthesis. 
A common route is to first identify a model and then compute an invariant set.
While intuitive, this approach is vulnerable to estimation bias, which can unexpectedly shift the synthesized RPI set and lead to unsafe control actions.


To mitigate this issue, recent research has studied data-driven RPI synthesis through model-set-based approaches, including zonotope-based reachability methods \cite{alanwar2023data,farjadniaRobustDataDrivenTubeBased2024} and set-membership methods \cite{chen2021data,mulagaleti2021data,wang2026data}.
Although they employ different set representations and computational tools, both approaches first characterize or over-approximate a family of models consistent with noisy data and then construct reachable or invariant sets that remain valid over this model family.
Zonotope-based methods provide efficient set-propagation mechanisms and have recently been embedded into tube-based predictive control with closed-loop guarantees \cite{farjadniaRobustDataDrivenTubeBased2024}.
Set-membership methods instead characterize the admissible model family through explicit consistency constraints and subsequently synthesize controllers or invariant sets that are robust over this family.
The method in \cite{mulagaleti2021data} jointly synthesizes models, feedback gains, and invariant sets, but relies on non-convex optimization and a preselected set topology, while the recent method in \cite{wang2026data} constructs data-driven RPI sets under process or measurement noise with polytopic or ellipsoidal bounds.
The conservatism of these model-set-based methods therefore depends strongly on the tightness of the data-consistent model set and the subsequent outer approximations.

More importantly, extending data-driven RPI synthesis to practical output-feedback scenarios remains challenging and comparatively less explored \cite{mayneRobustOutputFeedback2006}.
The coupling of process noise, measurement noise, and estimation errors introduces additional uncertainty channels into the invariant-set construction, making direct extensions of existing zonotopic and set-membership approaches nontrivial.


To break these bottlenecks, data-driven frameworks based on Willems' fundamental lemma \cite{willemsNotePersistencyExcitation2004,depersisFormulasDataDrivenControl2020} provide a powerful alternative.
By establishing that any valid system trajectory can be expressed as a linear combination of historical data matrices under a persistency of excitation condition, this lemma enables the direct parameterization of system behavior from raw data in noise-free scenarios. 
This capability allows the synthesis conditions to be written without explicitly constructing intermediate model proxies, thereby avoiding the potential conservatism associated with model-uncertainty descriptions.
This characteristic inspires the development of direct data-driven RPI synthesis methods, which express closed-loop synthesis conditions through measured data rather than through an identified model.

Motivated by these gaps, this paper develops a direct data-driven framework for RPI set synthesis, covering both input-state feedback and output-feedback cases.
In contrast to the methods mentioned above, the proposed conditions are written directly in terms of measured data matrices and solved as SDPs.
The main contributions of this paper are summarized as follows:

\begin{itemize}
    \item 
    We propose a direct data-driven method for RPI set synthesis in state-feedback and output-feedback scenarios under bounded process and measurement noise.
    Compared with recent data-driven RPI methods, the invariant sets are computed from data-dependent matrix inequalities without explicitly constructing an intermediate model-uncertainty set.

    \item 
    A decoupled synthesis framework is developed.
    This framework retains the nominal LQ-type gain-design objective while reducing the non-convex coupling between gain synthesis and RPI-set computation.

    \item 
    The framework is extended to the output-feedback case with offline data.
    By combining an estimation-error invariant set and a true-state RPI set, the proposed method provides data-driven invariant bounds for the physical state when the true state is unavailable online.

\end{itemize}

\textit{Notation:} 
$\mathbb{R}$ and $\mathbb{Z}$ denote the sets of real numbers and integers, respectively.
$\mathbb{Z}_{+}$ denotes the set of nonnegative integers.
$\mathbb{R}^n$ and $\mathbb{R}^{n\times m}$ denote the $n$-dimensional Euclidean space and the space of $n\times m$ real matrices, respectively.
For a matrix $M$, $M^T$, $M^\dagger$, $\operatorname{rank}(M)$, and $\operatorname{tr}(M)$ denote its transpose, Moore--Penrose pseudoinverse, rank, and trace, respectively.
For a symmetric matrix $P$, $P\succ0$ and $P\succeq0$ denote positive definiteness and positive semidefiniteness, respectively.
$\lambda_{\min}(P)$ and $\lambda_{\max}(P)$ denote the minimum and maximum eigenvalues of a symmetric matrix $P$.
$I_n$ is the $n\times n$ identity matrix, and $\mathbf{0}$ denotes a zero matrix with compatible dimensions.
The operator $\operatorname{diag}(\cdot)$ denotes a block-diagonal matrix constructed from its arguments.
For a positive definite matrix $Q$, $Q^{1/2}$ denotes its symmetric square root.
In a symmetric block matrix, $*$ denotes the transpose of the corresponding off-diagonal block.

%% file: body/2.Preliminary.tex
\section{Preliminaries and Problem Formulation}
\subsection{System Description}
Consider a discrete-time linear time-invariant system with time step $k\in \mathbb{Z}^+$:
\begin{subequations}\label{eq:system}
    \begin{align}
        x_{k+1} &= A x_k + B u_k + w_k, \label{eq:system:a}\\
        y_k &= C x_k + v_k  , \label{eq:system:b}
    \end{align}
\end{subequations}
where $x_k\in\mathbb{R}^{n}$, $y_k\in\mathbb{R}^{n_y}$, $u_k\in\mathbb{R}^{n_u}$, $w_k\in\mathbb{R}^{n}$, and $v_k\in\mathbb{R}^{n_y}$.
$\mathcal{X}$, $\mathcal{U}$, and $\mathcal{Y}$ are compact and convex sets representing the state, input, and output constraints, respectively.
The process noise $w_k$ and measurement noise $v_k$ are assumed to lie in the bounded ellipsoidal sets $\mathcal{W}$ and $\mathcal{V}$, respectively, i.e., $\mathcal{W} = \{ w \mid w^T Q_w^{-1} w \leq 1, Q_w \succ 0 \}$ and $\mathcal{V} = \{ v \mid v^T Q_v^{-1} v \leq 1, Q_v \succ 0 \}$.

Let the superscript $d$ denote data collected offline.
Consider the offline state, input, and output sequences $\{x^d_k\}_{k=0}^{L}$, $\{u^d_k\}_{k=0}^{L-1}$, and $\{y^d_k\}_{k=0}^{L-1}$ generated by \eqref{eq:system}, alongside the unknown process and measurement noise sequences $W_0 \triangleq [w^d_0, \cdots, w^d_{L-1}]$ and $V_0 \triangleq [v^d_0, \cdots, v^d_{L-1}]$.
The corresponding data matrices are constructed as $X_0 \triangleq [x^d_0, \cdots, x^d_{L-1}]$, $X_1 \triangleq [x^d_1, \cdots, x^d_{L}]$, $U_0 \triangleq [u^d_0, \cdots, u^d_{L-1}]$, and $Y_0 \triangleq [y^d_0, \cdots, y^d_{L-1}]$.

\begin{assumption}\label{Assumption:UnknownSys}
    The system matrices $A$, $B$, and $C$ are unknown, but the pairs $(A,B)$ and $(A,C)$ are controllable and observable, respectively.
    Only the offline data matrices $X_0, X_1, U_0, Y_0$ are accessible.
\end{assumption}

\begin{assumption}\label{Assumption:SNR}
    There exist known scalars $\gamma_w >0$ and $\gamma_v >0$ such that the signal-to-noise ratio (SNR) conditions $W_0 W_0^T  \preceq \gamma_w X_1 X_1^T , \ V_0 V_0^T \preceq \gamma_v Y_0 Y_0^T$ hold.
\end{assumption}

\begin{assumption}\label{Assumption:FullRank}
    The data matrix $[X_0\ U_0]^T$ has full row rank, i.e., $\text{rank}\begin{bmatrix} X_0 \\ U_0 \end{bmatrix} = n + n_u$.
    In disturbance-free settings, this rank condition is guaranteed by a persistently exciting input of order $n+1$ under the controllability assumption \cite{depersisFormulasDataDrivenControl2020}.
\end{assumption}

\begin{definition}\label{Definition:S}
    A set $\mathcal{S}$ is said to be an RPI set for system~\eqref{eq:system} if $x_{k+1} \in \mathcal{S}$ for all $x_k \in \mathcal{S}$ and all admissible disturbances in \eqref{eq:system}.
\end{definition}

\subsection{Problem Formulation}
In this paper, we investigate direct data-driven RPI synthesis under state-feedback and output-feedback scenarios for \eqref{eq:system}. 
In the proposed framework, we parameterize the RPI set as an ellipsoid with a free shape matrix.
This compact quadratic representation yields tractable invariance conditions and enables direct optimization of the invariant-set geometry.

The problem considered in this paper is hierarchical.
First, we investigate the idealized scenario where the full state vector $x_k$ is precisely measurable during both the offline data collection phase and the online phase (i.e., $C = I$ and $v_k = 0$).
Unlike \cite{wang2026data}, we retain an LQ-type feedback design while decoupling the feedback gain $K$ from the RPI set $\mathcal{S}_x$.
The problem is formalized as follows.
\begin{problem}\label{problem:1}
    For system~\eqref{eq:system:a} with fully measurable states, given offline data $X_0, X_1, U_0$ satisfying Assumption~\ref{Assumption:SNR}-\ref{Assumption:FullRank}, the objective is to determine a stabilizing state-feedback gain $K$ and an RPI set $\mathcal{S}_x$ directly from data such that $x_{k+1} \in \mathcal{S}_x$ for all $x_k \in \mathcal{S}_x$ and $w_k \in \mathcal{W}$.
\end{problem}

We next consider a practical output-feedback scenario in which only the input $u_k$ and the noise-corrupted output $y_k$ are measurable online, while the full state $x_k$ is available during offline data collection.
Such a setting is common in experimental platforms, where high-precision motion-capture systems can provide complete state information in the laboratory, but the deployed system must operate using partial state and noisy onboard measurements, such as GPS signals.

Since the true state is unavailable online, an observer-based feedback structure is introduced to characterize invariant bounds for the physical state.
Specifically, a Luenberger observer with gain $L$ is used to reconstruct the state estimate $\hat{x}_k$, and the control input is implemented as $u_k = K\hat{x}_k$.
By bounding the estimation error $e_k = x_k - \hat{x}_k$ with an invariant set $\mathcal{S}_e$, we obtain the true-state RPI set $\tilde{\mathcal{S}}_x$ to provide guarantees for the true state $x_k$ \cite{mayneRobustOutputFeedback2006}.
The problem is formalized as follows.
\begin{problem}\label{problem:2}
    For system~\eqref{eq:system} with states unavailable online, given offline data $X_0, X_1, U_0, Y_0$ satisfying Assumption~\ref{Assumption:UnknownSys}-\ref{Assumption:FullRank}, the objective is to determine a controller gain $K$, an observer gain $L$, an estimation-error RPI set $\mathcal{S}_e$, and a true-state RPI set $\tilde{\mathcal{S}}_x$ directly from data such that 
    \begin{equation}
        \begin{aligned}
            & e_{k+1} \in \mathcal{S}_e \text{ for all } e_k \in \mathcal{S}_e, w_k \in \mathcal{W}, \text{ and } v_k \in \mathcal{V},\\
            & x_{k+1} \in \tilde{\mathcal{S}}_x \text{ for all } x_k \in \tilde{\mathcal{S}}_x, e_k \in \mathcal{S}_e, w_k \in \mathcal{W}, \text{ and } v_k \in \mathcal{V}.
        \end{aligned}
    \end{equation}
\end{problem}

\subsection{Mathematical Preliminaries}
We state two mathematical lemmas which serve as fundamental tools for the subsequent analytical developments.

\begin{lemma}[Young's Inequality \cite{petersen1987stabilization}]\label{Lemma:Young}
    For arbitrary matrices $X$, $Y$, $F \succ 0$ and a scalar $\delta>0$, the inequality $X F Y^T + Y F X^T \preceq \delta XFX^T + \delta^{-1} Y F Y^T$ holds.
\end{lemma}

\begin{lemma}[S-Procedure \cite{boydLinearMatrixInequalities1994}]\label{Lemma:S-proc}
    Let $f, g_1, \cdots, g_m : \mathbb{R}^n \to \mathbb{R}$ ($m\in \mathbb{Z}^+$) be quadratic functions. The implication $g_1(z) \geq 0, \cdots, g_m(z) \geq 0 \implies f(z) \geq 0$ holds if there exist scalars $\alpha_i \geq 0$ ($i=1,\cdots,m$) such that $f(z) - \sum_{i=1}^m \alpha_i g_i(z) \geq 0$ for all $z\in \mathbb{R}^n$.
\end{lemma}

%% file: body/3.IS.tex
\section{State-Feedback Synthesis}
To address Problem~\ref{problem:1}, substituting $u_k=Kx_k$ into \eqref{eq:system:a} yields the closed-loop invariance condition $(A+BK)x_k+w_k\in\mathcal{S}_x$ for all $x_k\in\mathcal{S}_x$ and $w_k\in\mathcal{W}$.
By designing $K$ and $\mathcal{S}_x$ sequentially, we reformulate this intractable condition as convex SDPs.

\subsection{State-Feedback Gain Design}
We employ a quadratic Lyapunov function $V_f(x)=x^TPx$ satisfying the following discrete-time algebraic Riccati inequality (DARI) to ensure the LQ-type performance bound of the closed-loop system:
\begin{equation}\label{eq:Lyapunov Condition}
    (A+BK)^T P(A+BK)-P \preceq -(Q+K^T R K).
\end{equation}
Since the system pair $(A,B)$ is unknown, standard model-based solvers for the algebraic Riccati equation cannot be directly applied to obtain $P$ and $K$.
Instead, based on Assumption~\ref{Assumption:FullRank} and \cite[Theorem~34]{van2020data}, there exist basis matrices $\Phi_1^\dagger\in \mathbb{R}^{L\times n}$ and $\Phi_2^\dagger\in \mathbb{R}^{L\times n_u}$ such that
\begin{equation}\label{eq:Phi}
    \begin{bmatrix} X_0\\U_0 \end{bmatrix} \begin{bmatrix} \Phi_1^\dagger & \Phi_2^\dagger \end{bmatrix} = I.
\end{equation}
Using \eqref{eq:Phi} in \eqref{eq:system:a} yields
\begin{equation}
    A=(X_1-W_0)\Phi_1^\dagger,\quad B=(X_1-W_0)\Phi_2^\dagger.
\end{equation}
Consequently, the closed-loop state matrix can be parameterized by data as $A+BK=(X_1-W_0)(\Phi_1^\dagger+\Phi_2^\dagger K)$.
The components $\Phi_1^\dagger$ and $\Phi_2^\dagger$ can be obtained offline by solving the following SDP \cite{liu2024learning}:
\begin{equation}\label{eq:data_driven_LMI}
    \begin{aligned}
        &\min_{\rho,M_2,\Phi_1^\dagger,\Phi_2^\dagger} \rho\\
        \text{s.t.}\quad
        &\begin{bmatrix} X_0\\U_0 \end{bmatrix}
        \begin{bmatrix} \Phi_1^\dagger & \Phi_2^\dagger \end{bmatrix}=I\\
        &\begin{bmatrix} M_2 & (\Phi_1^\dagger)^T\\ \Phi_1^\dagger & I \end{bmatrix}\succeq0\\
        &\rho I-M_2\succeq0.
    \end{aligned}
\end{equation}

Condition \eqref{eq:Lyapunov Condition} implies the infinite-horizon cost bound $\sum_{k=0}^{\infty}(x_k^TQx_k+u_k^TRu_k)\leq x_0^TPx_0$, where $x_0$ is the initial state.
To eliminate the dependence on $x_0$, we minimize the trace of $P$ as an aggregate surrogate for the LQ-type performance bound.
Drawing upon the linear matrix inequality (LMI) construction in \cite[Theorem~1]{dengEventTriggeredRobustMPC2024}, we present the following proposition to compute $P$ and $K$ satisfying the condition in \eqref{eq:Lyapunov Condition}.

\begin{proposition}\label{proposition:K}
    Let $\bar{P}\in \mathbb{R}^{n\times n} \succ 0$ and $Z_P\in \mathbb{R}^{n\times n}$ be symmetric matrices, and let $F_1\in \mathbb{R}^{L\times n}$ and $\bar{H}\in \mathbb{R}^{n\times n}\succ 0$ be decision variables.
    For a given scalar $\delta_1>0$, if the following SDP is feasible, then the recovered matrices $P$ and $K$ satisfy the condition in \eqref{eq:Lyapunov Condition}.
    \begin{subequations}
        \begin{align}
            \min_{\bar{P},F_1,\bar{H},Z_P} \quad & \operatorname{trace}(Z_P)\\
            \text{s.t.}\quad
            &\begin{bmatrix}
                \bar{P}-\gamma_w(1+\delta_1^{-1})X_1X_1^T & X_1F_1\\
                F_1^TX_1^T & (1+\delta_1)^{-1}\bar{H}
            \end{bmatrix}\succeq0\label{eq:SDP_1}\\
            &\begin{bmatrix}
                \bar{H} & (U_0F_1)^T & \bar{H}\\
                U_0F_1 & R^{-1} & 0\\
                \bar{H} & 0 & Q^{-1}
            \end{bmatrix}\succeq0\label{eq:SDP_2}\\
            &\begin{bmatrix}
                I & F_1\\
                F_1^T & \bar{H}
            \end{bmatrix}\succeq0\label{eq:SDP_3}\\
            &\bar{H}\succeq2\bar{P}\label{eq:SDP_4}\\
            &X_0F_1=\bar{H},\label{eq:SDP_5}\\
            &\begin{bmatrix}
                Z_P & I\\
                I & \bar{P}
            \end{bmatrix}\succeq0.\label{eq:SDP_6}
        \end{align}
    \end{subequations}
    where $P=\bar{P}^{-1}$ and $K=U_0F_1\bar{H}^{-1}$.
\end{proposition}

\begin{proof}
    Applying the Schur complement to Eq.~\eqref{eq:SDP_1} and using Assumption~\ref{Assumption:SNR}, we obtain
    \begin{equation}
        \bar{P}-(1+\delta_1^{-1})W_0W_0^T-(1+\delta_1)X_1F_1\bar{H}^{-1}F_1^TX_1^T\succeq0.
    \end{equation}
    From Eq.~\eqref{eq:SDP_5}, it follows that $(X_1-W_0)F_1\bar{H}^{-1} = (AX_0 + BU_0)F_1\bar{H}^{-1} = A+BK$.
    By leveraging Eq.~\eqref{eq:SDP_3} and Lemma~\ref{Lemma:Young}, we derive
    \begin{equation}\label{eq:LMI_1}
        \begin{aligned}
            &\bar{P}-(X_1-W_0)F_1\bar{H}^{-1} \bar{H} \bar{H}^{-1}F_1^T(X_1-W_0)^T\succeq 0
        \end{aligned}
    \end{equation}
    Similarly, applying the Schur complement to Eq.~\eqref{eq:SDP_2} and combining the result with Eq.~\eqref{eq:SDP_4}, we obtain
    \begin{equation}\label{eq:LMI_2}
        \begin{aligned}
            & \bar{P}^{-1} - \bar{H}^{-1} \succeq \bar{H}^{-1}\succeq Q+(U_0F_1\bar{H}^{-1})^TRU_0F_1\bar{H}^{-1}.
        \end{aligned}
    \end{equation}
    Substituting \eqref{eq:LMI_2} into \eqref{eq:LMI_1} and using the Schur complement, we obtain
    \begin{equation}
        (A+BK)^T\bar{P}^{-1}(A+BK)-\bar{P}^{-1}\preceq-(Q+K^TRK).
    \end{equation}
    Letting $P=\bar{P}^{-1}$ confirms that the condition in Eq.~\eqref{eq:Lyapunov Condition} is satisfied.
    Finally, applying the Schur complement to Eq.~\eqref{eq:SDP_6} gives
    \begin{equation}
        Z_P-\bar{P}^{-1}\succeq0.
    \end{equation}
    Therefore, $Z_P\succeq P$ and $\operatorname{trace}(Z_P)\geq\operatorname{trace}(P)$.
    For any fixed feasible $\bar{P}$, $F_1$, and $\bar{H}$, choosing $Z_P=\bar{P}^{-1}=P$ satisfies Eq.~\eqref{eq:SDP_6} and attains the minimum with respect to $Z_P$.
    Hence, minimizing $\operatorname{trace}(Z_P)$ is equivalent to minimizing $\operatorname{trace}(P)$ over the feasible controller designs.
\end{proof}

\begin{remark}\label{rmk:hyperparameter}
    The scalar $\delta_1>0$ provides an algebraic degree of freedom, which originates from the application of Lemma~\ref{Lemma:Young}.
    Since extreme values render the LMI infeasible by disproportionately magnifying performance penalties or uncertainty bounds, the optimal value of $\delta_1$ lies strictly within a finite practical range.
    Consequently, $\delta_1$ and subsequent hyperparameters can be efficiently determined via a 1-D logarithmic grid search.
\end{remark}

\subsection{Computation of the RPI Set}
In this section, we focus on the computation of the RPI set $\mathcal{S}_x$ when the full state vector $x_k$ is measurable.
The RPI set is parameterized as an ellipsoid:
\begin{equation}
    \mathcal{S}_x=\{x\in\mathbb{R}^n\mid x^TQ_S^{-1}x\leq1\},
\end{equation}
where $Q_S\succ0$ is the decision variable.
To reduce the ellipsoid size, we minimize $\operatorname{trace}(Q_S)$ as a convex size surrogate.

Note that the feedback gain $K$ is computed first via Proposition~\ref{proposition:K}.
This sequential design may lose some optimality in minimizing the RPI set size compared with joint synthesis.
However, it provides practical advantages: it retains the LQ-type gain and avoids the non-convex coupling between gain synthesis and invariant-set computation \cite{chen2021data,mulagaleti2021data}.
Moreover, with $K$ fixed, the RPI set can be optimized directly through its ellipsoidal shape matrix, without introducing fixed polyhedral templates or preselected supporting directions.
Based on this structure, we establish the following result.

\begin{theorem}\label{theorem:QS}
    For a given scalar $\alpha\in(0,1)$, if the following SDP is feasible, then its solution $Q_S\succ 0$ defines an RPI set $\mathcal S_x$.
    \begin{subequations}
        \begin{align}
            \min_{Q_S,\lambda}\quad &\operatorname{trace}(Q_S)\\
            \text{s.t.}\quad
            &\begin{bmatrix}
                (1-\alpha)Q_S & 0 & F_2^TX_1^T & F_2^T\\
                0 & \alpha Q_w^{-1} & I & 0\\
                X_1F_2 & I & Q_S-\lambda\gamma_wX_1X_1^T & 0\\
                F_2 & 0 & 0 & \lambda I
            \end{bmatrix}\succeq0\label{eq:SDP2_1}\\
            &F_2=(\Phi_1^\dagger+\Phi_2^\dagger K)Q_S\label{eq:SDP2_2}
        \end{align}
    \end{subequations}
    where $\lambda>0$ is an auxiliary decision variable.
\end{theorem}
\begin{proof}
    By applying the Schur complement to Eq.~\eqref{eq:SDP2_1} and incorporating Assumption~\ref{Assumption:SNR} and Lemma~\ref{Lemma:Young}, we obtain
    \begin{equation}
        \begin{bmatrix}
            (1-\alpha)Q_S & 0 & F_2^T(X_1-W_0)^T\\
            0 & \alpha Q_w^{-1} & I\\
            (X_1-W_0)F_2 & I & Q_S
        \end{bmatrix}\succeq0.
    \end{equation}
    By substituting Eq.~\eqref{eq:SDP2_2} and applying the Schur complement, we obtain
    \begin{equation}
        \begin{bmatrix}\Gamma\\I\end{bmatrix}
        Q_S^{-1}
        \begin{bmatrix}\Gamma\\I\end{bmatrix}^T
        -
        \begin{bmatrix}
            (1-\alpha)Q_S^{-1} & 0\\
            0 & \alpha Q_w^{-1}
        \end{bmatrix}\preceq0,
    \end{equation}
    where $\Gamma^T=(X_1-W_0)(\Phi_1^\dagger+\Phi_2^\dagger K)$.
    Pre- and post-multiplying both sides of the inequality by $\begin{bmatrix}x_k^T&w_k^T\end{bmatrix}$ and $\begin{bmatrix}x_k^T&w_k^T\end{bmatrix}^T$, respectively, yields
    \begin{equation}
        \begin{aligned}
            &\begin{bmatrix}x_k\\w_k\end{bmatrix}^T
            \left(
            \begin{bmatrix}(A+BK)^T\\I\end{bmatrix}
            Q_S^{-1}
            \begin{bmatrix}A+BK&I\end{bmatrix}\right.\\
            &\left.-
            \begin{bmatrix}
                (1-\alpha)Q_S^{-1} & 0\\
                0 & \alpha Q_w^{-1}
            \end{bmatrix}
            \right)
            \begin{bmatrix}x_k\\w_k\end{bmatrix}\preceq0.
        \end{aligned}
    \end{equation}
    This implies
    \begin{equation}
        x_{k+1}^TQ_S^{-1}x_{k+1}-(1-\alpha)x_k^TQ_S^{-1}x_k-\alpha w_k^TQ_w^{-1}w_k\leq0.
    \end{equation}
    From Lemma~\ref{Lemma:S-proc}, it follows that $x_{k+1}^TQ_S^{-1}x_{k+1}\leq1$ holds for all $x_k^TQ_S^{-1}x_k\leq1$ and $w_k^TQ_w^{-1}w_k\leq1$, confirming that $\mathcal{S}_x$ is an RPI set.
\end{proof}

%% file: body/4.IO.tex
\section{Output-Feedback Synthesis}
In this section, we extend the data-driven RPI synthesis framework to the output-feedback scenario.
By defining the estimation error as $e_k=x_k-\hat{x}_k$ and employing a Luenberger observer $\hat{x}_{k+1} = A\hat{x}_k + Bu_k + L(y_k - C\hat{x}_k)$ with the control law $u_k=K\hat{x}_k$, we obtain the augmented closed-loop dynamics
\begin{equation}\label{eq:xe}
    \begin{bmatrix} x_{k+1}\\ e_{k+1}  \end{bmatrix} = \begin{bmatrix} A + BK & -BK \\ 0 & A - LC \end{bmatrix} \begin{bmatrix} x_k\\e_k  \end{bmatrix} + \begin{bmatrix} I & 0 \\ I & -L \end{bmatrix} \begin{bmatrix} w_k \\ v_k     \end{bmatrix}.
\end{equation}
In \eqref{eq:xe}, the true state dynamics are intrinsically driven by the estimation error $e_k$. 
To manage this dependency, we adopt the robust output-feedback paradigm \cite{mayneRobustOutputFeedback2006} to systematically decouple the synthesis: we first compute a robust invariant set $\mathcal{S}_e$ for the error dynamics, and subsequently determine the overall RPI set $\tilde{\mathcal{S}}_x$ by treating $e_k \in \mathcal{S}_e$ as a bounded internal disturbance.
\begin{remark}
    The augmented closed-loop dynamics \eqref{eq:xe} are used to describe the observer-based closed-loop architecture and the associated estimation-error dynamics, which are subsequently used to synthesize the observer gain and invariant bounds.
    In the proposed formulation, the matrices $\{A,B,C\}$ are not explicitly identified, and the corresponding synthesis conditions are constructed directly from the available offline data.
\end{remark}

To facilitate the subsequent data-driven LMI synthesis, the following lemma is established to provide a joint characterization of the noise matrices $W_0$ and $V_0$.

\begin{lemma}\label{Lemma:SNR_all}
    Under Assumption~\ref{Assumption:SNR}, for any scalar $\epsilon > 0$, the joint noise matrix satisfies the following block-diagonal data-driven upper bound:
    \begin{equation}
        \begin{aligned}
            \begin{bmatrix} W_0 \\ V_0 \end{bmatrix} \begin{bmatrix} W_0 \\ V_0 \end{bmatrix}^T \preceq 
            \begin{bmatrix} \gamma_w(1+\epsilon)X_1 X_1^T & 0 \\ 0 & \gamma_v(1+\epsilon^{-1}) Y_0 Y_0^T \end{bmatrix}.
        \end{aligned}
    \end{equation}
\end{lemma}
\begin{proof}
    The bound trivially holds by applying Lemma~\ref{Lemma:Young} to the cross terms $W_0 V_0^T$ and $V_0 W_0^T$.
\end{proof}

The scalar $\epsilon$ can be selected by minimizing the $\epsilon$-dependent part
\begin{equation}
    \epsilon \gamma_w \operatorname{tr} \left( X_1 X_1^T  \right) + \epsilon^{-1} \gamma_v \operatorname{tr} \left( Y_0 Y_0^T \right),
\end{equation}
which yields
\begin{equation}
\epsilon^\star = \sqrt{\frac{\gamma_v \operatorname{tr}(Y_0Y_0^T)}{\gamma_w \operatorname{tr}(X_1X_1^T)}}.
\end{equation}
This selection balances the relative scales of the diagonal blocks in Lemma~\ref{Lemma:SNR_all} and minimizes the trace of the $\epsilon$-dependent upper bound.


\subsection{Design of Controller and Observer Gains}
To avoid the non-convexities associated with concurrent output-feedback synthesis in \eqref{eq:xe}, we employ a sequential design strategy. 
With the data-driven state-feedback gain $K$ predetermined in Proposition~\ref{proposition:K}, we synthesize an observer gain $L$ that robustly stabilizes the error dynamics $e_{k+1} = (A-LC)e_k + w_k - Lv_k$ while reducing the worst-case estimation-error bound.

For the error dynamics, we seek a positive definite matrix $P_e$ and an observer gain $L$ such that the following condition holds:
\begin{equation}\label{eq:observer}
    (A-LC) P_e (A-LC)^T - P_e \preceq -(Q_e + L R_e L^T).
\end{equation}
Although \eqref{eq:observer} is algebraically dual to the state-feedback Lyapunov condition \eqref{eq:Lyapunov Condition}, it can also be interpreted as a covariance-type condition commonly used in robust Kalman filtering \cite{xie1994robust}.
Here, $Q_e \succ 0$ and $R_e \succ 0$ specify the relative weights assigned to the process and measurement noise contributions, respectively.
Accordingly, reducing the trace of $P_e$ provides a surrogate for reducing the guaranteed estimation-error bound across all state components.


The nominal closed-loop stability still follows from the standard separation principle.
However, the estimation error enters the physical-state dynamics through $-BK e_k$, and therefore the separate designs of $K$ and $L$ are generally suboptimal for minimizing the resulting invariant bounds.
This sequential structure should be interpreted as a practical trade-off between convex synthesis and global optimality.

Using \eqref{eq:Phi} in \eqref{eq:system:b} yields $C = (Y_0 - V_0) \Phi_1^\dagger$.
The closed-loop matrix in \eqref{eq:observer} can then be parameterized as
\begin{equation}\label{eq:A-LC}
    (A-LC)=(X_1-LY_0)\Phi_1^\dagger + (LV_0-W_0)\Phi_1^\dagger.
\end{equation}
Based on this data-driven parameterization, we establish the following theorem.
\begin{theorem}\label{theo:L}
    For a given scalar $\delta_2 >0$, the observer gain $L$ and the matrix $P_e$ satisfying Eq.~\eqref{eq:observer} can be obtained if there exist symmetric matrices $\bar{P}_e \in \mathbb{R}^{n\times n}\succ 0$ and $Z_e\in \mathbb{R}^{n\times n}$, and a matrix $F_e\in \mathbb{R}^{n\times n_y}$ such that the following SDP is feasible:
    \begin{subequations}\label{eq:L}
        \begin{align}
            \min_{\bar{P}_e, F_e,Z_e} \quad &  \text{trace}(Z_e) \\
            s.t. \quad 
            & \begin{bmatrix}
                    \Xi_{11} & \Xi_{12} & \Xi_{13} & \Xi_{14} & \Xi_{15} & \Xi_{16} \\
                    * & \Xi_{22} & 0 & 0 & 0 & 0 \\
                    * & * & \Xi_{33} & 0 & 0 & 0 \\
                    * & * & * & I & 0 & 0 \\
                    * & * & * & * & I & 0 \\
                    * & * & * & * & * & I
                \end{bmatrix} \succeq 0,  \label{eq:L1} \\
            &\begin{bmatrix}
                I & \Phi_1^\dagger\\
                (\Phi_1^\dagger)^T & \bar{P}_e
            \end{bmatrix} \succeq 0 , \label{eq:L2}\\
            &\begin{bmatrix}
                Z_e & I\\
                I & \bar{P}_e
            \end{bmatrix} \succeq 0. \label{eq:L3}
        \end{align}
    \end{subequations}
    where
    \begin{equation*}
        \begin{aligned}
            &\Xi_{11} = \bar{P}_e, \quad \Xi_{12} =  (\bar{P}_eX_1-F_eY_0) \Phi_1^\dagger,\\
            &\Xi_{13} = F_e Y_0, \quad \Xi_{14} = \bar{P}_e Q_e^{1/2}, \quad \Xi_{15} = F_e R_e^{1/2},\\
            &\Xi_{16} = \sqrt{(1+\delta_2^{-1})\gamma_w (1+\epsilon)} \bar{P}_e X_1, \quad \Xi_{22} = (1+\delta_2)^{-1}\bar{P}_e,\\
            &\Xi_{33} = (1+\delta_2^{-1})^{-1} \gamma_v^{-1} (1+\epsilon^{-1})^{-1}I.
        \end{aligned}
    \end{equation*}
    Then, the observer parameters are recovered as $P_e = \bar{P}_e^{-1}$ and $L =\bar{P}_e^{-1} F_e $.
\end{theorem}

\begin{proof}
    Taking the Schur complement of the identity blocks in the LMI and applying the congruence transformation with $\operatorname{diag}\{\bar{P}_e^{-1}, I, I\}$ yield
    \begin{equation}
        \begin{bmatrix}
            \Psi &  (X_1-LY_0) \Phi_1^\dagger  & LY_0\\
            * & (1+\delta_2)^{-1}\bar{P}_e &0\\
            * & * & (1+\delta_2^{-1})^{-1} \gamma_v^{-1} (1+\frac{1}{\epsilon})^{-1}I
        \end{bmatrix} \succeq 0,
    \end{equation}
    where $\Psi = \bar{P}^{-1}_e - Q_e - \bar{P}^{-1}_e F_e R_e F_e^T \bar{P}^{-1}_e -(1+\delta_2^{-1})\gamma_w (1+\epsilon) X_1X_1^T$.
    By applying the Schur complement again and using Lemma~\ref{Lemma:SNR_all} and Lemma~\ref{Lemma:Young}, we obtain
    \begin{equation}
        \begin{aligned}
            & - \begin{bmatrix} I \\-L^T  \end{bmatrix}^T \begin{bmatrix} X_1-W_0 \\Y_0-V_0\end{bmatrix} \Phi_1^\dagger \bar{P}^{-1}_e (\Phi_1^\dagger)^T \begin{bmatrix} X_1-W_0 \\Y_0-V_0\end{bmatrix}^T \begin{bmatrix} I \\-L^T  \end{bmatrix}\\
            &+ \bar{P}^{-1}_e - (Q_e + L R_e L^T) \succeq 0,
        \end{aligned}
    \end{equation}
    which is equivalent to
    \begin{equation}
        \bar{P}^{-1}_e - (Q_e + L R_e L^T) - (A-LC) \bar{P}^{-1}_e (A-LC)^T \succeq 0.
    \end{equation}
    By setting $P_e=\bar{P}_e^{-1}$, we obtain
    \begin{equation}
        (A-LC)P_e(A-LC)^T - P_e \preceq -(Q_e + L R_e L^T),
    \end{equation}
    which proves \eqref{eq:observer}.
    Finally, applying the Schur complement to Eq.~\eqref{eq:L3} gives $Z_e-\bar{P}_e^{-1}\succeq0$.
    Therefore, $Z_e\succeq P_e$ and $\operatorname{trace}(Z_e)\geq\operatorname{trace}(P_e)$.
    For any fixed feasible $\bar{P}_e$ and $F_e$, choosing $Z_e=\bar{P}_e^{-1}=P_e$ satisfies Eq.~\eqref{eq:L3} and attains the minimum with respect to $Z_e$.
    Hence, minimizing $\operatorname{trace}(Z_e)$ is equivalent to minimizing $\operatorname{trace}(P_e)$ over the feasible observer designs.
\end{proof}

\subsection{Characterization of the Estimation-Error RPI Set}
To solve Problem~\ref{problem:2}, we first consider the invariant set of the error dynamics in \eqref{eq:xe}, which are driven by both process and measurement noise.
Let $\mathcal{S}_e = \{e \mid e^T Q_{S_e}^{-1} e \leq 1\}$ be the RPI set of the estimation error $e_k$.
Then we propose the following theorem.

\begin{theorem}\label{theo:QS_e}
    Given scalars $0 < \alpha_1, \alpha_2,\alpha_3 < 1$ with $\alpha_1 + \alpha_2 + \alpha_3 = 1$, if there exist scalars $\lambda_1 >0$ and $\lambda_2 >0$ such that the following SDP is feasible, then $Q_{S_e}\succ 0$ characterizes the RPI set of the estimation error $e_k$:
    \setlength{\arraycolsep}{2.5pt}
    \begin{equation}\label{eq:QS_e}
        \begin{aligned}
            &\min_{Q_{S_e}, \lambda_1, \lambda_2}  \text{trace}(Q_{S_e})\\
            &\begin{bmatrix} 
                \alpha_1 Q_{S_e}  &0&0& \Lambda_{14} &Q_{S_e}(\Phi_1^\dagger)^T & Q_{S_e}(\Phi_1^\dagger)^T\\
                *&\alpha_2 Q_{w}^{-1} &0 &I & 0 & 0\\
                *&*& \alpha_3 Q_{v}^{-1} &-L^T & 0 & 0\\
                * &*&*& \Lambda_{44} & 0 & 0\\
                * & * &*&*&\lambda_2 I & 0\\
                * &* &* &*&*&\lambda_1 I
            \end{bmatrix}\succeq 0
        \end{aligned}
    \end{equation}
    where
    \begin{equation*}
        \begin{aligned}
            &\Lambda_{14} = Q_{S_e}[(X_1-LY_0)\Phi_1^\dagger] ^T, \\
            &\Lambda_{44} = Q_{S_e}-\lambda_1 \gamma_w (1+\epsilon) X_1X_1^T-\lambda_2 \gamma_v (1+\frac{1}{\epsilon}) LY_0Y_0^TL^T.
        \end{aligned}
    \end{equation*}
\end{theorem}
\begin{proof}
    Multiplying both sides by $diag([Q^{-1}_{S_e},I,I,I,I,I])$ and applying the Schur complement yields
    \begin{equation}
        \begin{bmatrix} 
            \Lambda_{11} &0&0&[(X_1-LY_0)\Phi_1^\dagger] ^T\\
            *&\alpha_2 Q_{w}^{-1} &0 &I\\
            *&*& \alpha_3 Q_{v}^{-1} &-L^T\\
            * &*&*& \Lambda_{44}
        \end{bmatrix}\succeq 0,
    \end{equation}
    where
    $\Lambda_{11} = \alpha_1 Q_{S_e}^{-1}-\lambda_1^{-1} (\Phi_1^\dagger)^T \Phi_1^\dagger -\lambda_2^{-1}(\Phi_1^\dagger)^T \Phi_1^\dagger$ and
    $\Lambda_{44} = Q_{S_e}-\lambda_1 \gamma_w (1+\epsilon) X_1X_1^T-\lambda_2 \gamma_v (1+\frac{1}{\epsilon}) LY_0Y_0^TL^T$.
    Using Lemma~\ref{Lemma:Young}, Lemma~\ref{Lemma:SNR_all}, and \eqref{eq:A-LC}, we obtain
    \begin{equation}
        \begin{aligned}
            \begin{bmatrix} 
            \alpha_1 Q_{S_e}^{-1} &0&0& [A-LC] ^T\\
            *&\alpha_2 Q_{w}^{-1} &0 &I\\
            *&*& \alpha_3 Q_{v}^{-1} &-L^T\\
            *&*&*& Q_{S_e}
            \end{bmatrix}
            \succeq 0.
        \end{aligned}
    \end{equation}
    We apply the Schur complement again and then pre- and post-multiply both sides of the inequality by $\begin{bmatrix}e_k^T&w_k^T&v_k^T\end{bmatrix}$ and $\begin{bmatrix}e_k^T&w_k^T&v_k^T\end{bmatrix}^T$, respectively, to obtain
    \begin{equation}
        e_{k+1}^T Q_{S_e}^{-1} e_{k+1} - \alpha_1 e_{k}^T Q_{S_e}^{-1} e_{k} - \alpha_2 w_k^T Q_w^{-1} w_k -\alpha_3 v_k^T Q_v^{-1} v_k \leq 0.
    \end{equation} 
    From Lemma~\ref{Lemma:S-proc}, it follows that $e_{k+1}\in \mathcal{S}_e$ for all $e_k\in \mathcal{S}_e$, $w_k\in \mathcal{W}$, and $v_k\in \mathcal{V}$.
\end{proof}

\subsection{Characterization of the State RPI Set}
Based on the above discussion, let $\tilde{\mathcal{S}}_x = \{ x \mid x^T Q_{\tilde{S}}^{-1} x \leq 1 \}$ denote an RPI set for the true state $x_k$ in system \eqref{eq:xe}.
Then we propose the following theorem.

\begin{theorem}
    Given scalars $0 < \beta_1, \beta_2,\beta_3 < 1$ satisfying $\beta_1 + \beta_2 + \beta_3 = 1$, if there exist scalars $\mu_1 >0$ and $\mu_2 >0$ such that the following SDP is feasible, then $Q_{\tilde{S}}\succ 0$ characterizes the RPI set $\tilde{\mathcal{S}}_x$:
    \setlength{\arraycolsep}{2.5pt}
    \begin{equation}\label{eq:Q_tilde_S}
        \begin{aligned}
            &\min_{Q_{\tilde{S}}, \mu_1, \mu_2}  \text{trace}(Q_{\tilde{S}})\\
            s.t. \quad
            & \begin{bmatrix}
                \beta_1 Q_{\tilde{S}} & 0 & 0 & \Upsilon_{14} & \Upsilon_{15} & 0 \\
                * & \beta_2 Q_{S_e}^{-1} & 0 & \Upsilon_{24} & 0 & (\Phi_2^\dagger K)^T \\
                * & * & \beta_3 Q_w^{-1} & I & 0 & 0 \\
                * & * & * & \Upsilon_{44} & 0 & 0 \\
                * & * & * & * & \mu_1 I & 0 \\
                * & * & * & * & * & \mu_2 I
            \end{bmatrix} \succeq 0
        \end{aligned}
    \end{equation}
    where
    \begin{equation*}
        \begin{aligned}
            &\Upsilon_{14} = Q_{\tilde{S}}(X_1(\Phi_1^\dagger +\Phi_2^\dagger K))^T,\ \Upsilon_{15} = Q_{\tilde{S}}(\Phi_1^\dagger +\Phi_2^\dagger K)^T,\\
            &\Upsilon_{24} = (-X_1\Phi_2^\dagger K)^T,\ \Upsilon_{44} = Q_{\tilde{S}} - \gamma_w (1+\epsilon) (\mu_1+\mu_2) X_1X_1^T.
        \end{aligned}
    \end{equation*}
\end{theorem}
\begin{proof}

    We apply the Schur complement to \eqref{eq:Q_tilde_S}, pre- and post-multiply the resulting matrix by $\text{diag}(Q_{\tilde{S}}^{-1}, I, I, I, I)$, and repeatedly apply the Schur complement to obtain
    \begin{equation}
        \begin{aligned}
            &\begin{bmatrix} 
                \beta_1 Q_{\tilde{S}}^{-1} &0&0 & Q_{\tilde{S}}^{-1}\Upsilon_{14} \\ 
                *&\beta_2Q_{S_e}^{-1} &0 & (-X_1\Phi_2^\dagger K)^T\\
                *&*&\beta_3 Q_w^{-1} &I\\
                * & * & * & Q_{\tilde{S}}
            \end{bmatrix} \\
            -&
            \mu_1^{-1} \Psi_1 \Psi_1^T - \mu_2^{-1}\Psi_2 \Psi_2^T -\gamma_w (1+\epsilon) (\mu_1+\mu_2) \Psi_3 \Psi_3^T \succeq 0
        \end{aligned}
    \end{equation}
    where $\Psi_1 = [\Phi_1^\dagger +\Phi_2^\dagger K\ 0\ 0\ 0]^T$, $\Psi_2 = [0\ \Phi_2^\dagger K\ 0\ 0]^T$, and $\Psi_3 = [0\ 0\ 0\ X_1^T]^T$. 
    Based on Lemma~\ref{Lemma:SNR_all} and Lemma~\ref{Lemma:Young}, we obtain
    \begin{equation}
        \begin{bmatrix} 
            \beta_1 Q_{\tilde{S}}^{-1}&0&0 & (A+BK)^T \\ 
            *&\beta_2Q_{S_e}^{-1}&0 & (-BK)^T\\
            *&*&\beta_3 Q_w^{-1} &I\\
            * & * & * & Q_{\tilde{S}}
        \end{bmatrix}\succeq 0.
    \end{equation}
    We apply the Schur complement again and then pre- and post-multiply this matrix inequality by $[x_k^T, e_k^T, w_k^T]$ and its transpose, respectively, to obtain
    \begin{equation}
        \begin{aligned}
            & x^T_{k+1} Q_{\tilde{S}}^{-1} x_{k+1} -1 - \beta_1(x^T_k Q_{\tilde{S}}^{-1} x_k -1)\\
            -&\beta_2(e^T_k Q_{S_e}^{-1} e_k -1) - \beta_3(w_k^T Q_w^{-1} w_k - 1) \leq 0.
        \end{aligned}
    \end{equation}
    Based on the S-procedure, $x_{k+1}^T Q_{\tilde{S}}^{-1} x_{k+1} \leq 1$ holds for all $x_k^T Q_{\tilde{S}}^{-1} x_k \leq 1$, $e_k^T Q_{S_e}^{-1} e_k \leq 1$, and $w_k^T Q_w^{-1} w_k \leq 1$.
\end{proof}

%% file: body/5.Sim.tex
\section{Simulation}
In this section, we present the simulation results of the proposed control design for both state-feedback and output-feedback control scenarios.

\subsection{Case Study I: State-Feedback Control}
In this part, we investigate the scenario where the full state of the system is available for feedback.
We consider the system dynamics given by Alanwar et al.~\cite{alanwar2022data}:
\begin{equation}\label{eq:study_1}
    A = \begin{bmatrix} 0.9455 & -0.2426 \\ 0.2486 & 0.9455 \end{bmatrix},\ 
    B = \begin{bmatrix} 0.1 \\ 0 \end{bmatrix}.
\end{equation}
The disturbance matrix is set to $Q_w^{-1} = \operatorname{diag}(2500, 2500)$.
With $\gamma_w = 2.188\times 10^{-3}$, $Q=10^{-2}$, and $R=10^{-4}$, we collect an offline trajectory of length $T=200$.
Following Remark~\ref{rmk:hyperparameter}, 1-D grid searches over practical intervals yield the tuning hyperparameters $\delta_1=0.9$ and $\alpha=0.3$. 
Given these scalars, the state-feedback control law $K = \begin{bmatrix} -9.0896 & -8.2980 \end{bmatrix}$ and the bounding matrix $Q_S = \begin{bmatrix} 0.0169 & -0.0085 \\ -0.0085 & 0.0073 \end{bmatrix}$ are obtained from Proposition~\ref{proposition:K} and Theorem~\ref{theorem:QS}. 
The comparison between the proposed method and the model-based minimal robust positively invariant (mRPI) set \cite{rakovic2005invariant}, i.e., the smallest RPI set for the known closed-loop dynamics, is shown in Fig.~\ref{fig:P1C1}.
In this experiment, the proposed data-driven RPI set contains the model-based mRPI set, and their areas are 0.0225 and 0.0118, respectively.

Next, we benchmark our approach against the method in \cite{mulagaleti2021data}, with each framework using its own synthesized feedback gain and invariant set.
Because the method in \cite{mulagaleti2021data} determines noise bounds through implicit optimization, we evaluate the performance under both ellipsoidal ($Q_w^{-1}$) and amplitude-bounded ($\|w(t)\|_{\infty} \leq 0.0141$) noise configurations to ensure a fair comparison.
The result is shown in Fig.~\ref{fig:state_feedback_data_eval}.
Since the proposed formulation uses ellipsoidal disturbance bounds, the disturbances are handled through ellipsoidal outer approximations in amplitude-bounded scenarios, which may introduce additional geometric conservatism.
In the ellipsoidal case, the areas of the proposed and comparison sets are 0.0225 and 0.0198, respectively.
The proposed formulation avoids the coupling between model selection, feedback gains, and invariant sets in \cite{mulagaleti2021data}.

\begin{figure}[htbp]
    \centering
    \includegraphics[width=0.68\linewidth]{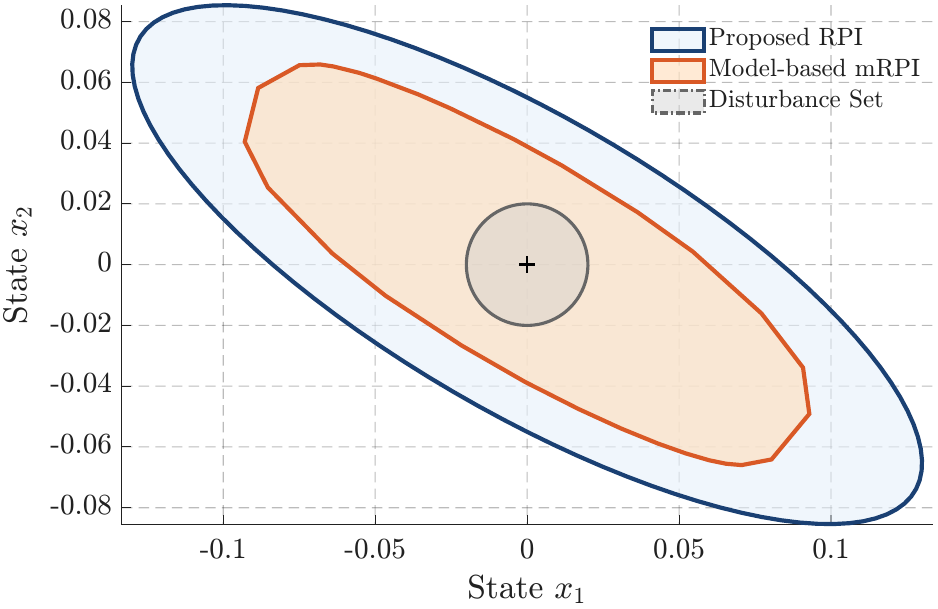}
    \caption{Evaluation of the proposed state-feedback RPI set against the model-based mRPI set.}
    \label{fig:P1C1}
\end{figure}

\begin{figure}[htbp]
    \centering
    \includegraphics[width=0.68\linewidth]{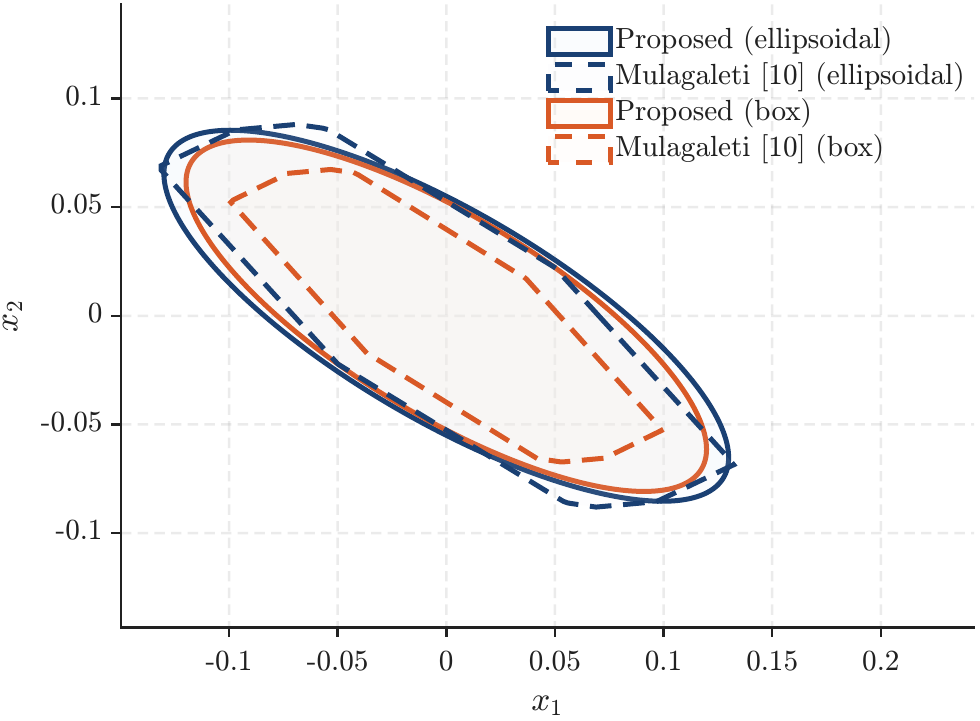}
    \caption{Evaluation of the proposed state-feedback RPI set against the data-driven set-membership method in \cite{mulagaleti2021data}.}
    \label{fig:state_feedback_data_eval}
\end{figure}

To complement the comparison in Fig.~\ref{fig:state_feedback_data_eval}, Table~\ref{table:perf_comparison} reports a comparison with a zonotopic baseline based on \cite{farjadniaRobustDataDrivenTubeBased2024,alanwar2023data} and the set-membership method \cite{mulagaleti2021data} for different ratios of $Q$ and $R$.
The proposed SDP and the method in \cite{mulagaleti2021data} use the same pre-designed feedback gain $K$, whereas the hard-constraint zonotopic baseline computes its gain from a data-consistent nominal model.
The hard-constraint zonotopic baseline follows the tube construction in \cite{farjadniaRobustDataDrivenTubeBased2024}, while the admissible model set is further tightened following Proposition~4 in \cite{alanwar2023data} to reduce conservatism.
Under ellipsoidal noise, the proposed SDP gives the smallest coordinate-wise bounds in all three cases.
Under amplitude-bounded noise, the hard-constraint zonotope gives the smallest $|x_1|$ bound in all three cases and the smallest $|x_2|$ bound in Case 2, while the method in \cite{mulagaleti2021data} gives the smallest $|x_2|$ bound in Cases 1 and 3.

\begin{table}[htbp] \centering \caption{Performance Comparison under Various Noise Geometries and Weight Settings ($Q=10^{-2}$).} \label{table:perf_comparison} \renewcommand{\arraystretch}{1.0} \resizebox{\columnwidth}{!}{\begin{tabular}{@{}llcccc@{}} \toprule \multirow{2}{*}{\textbf{Settings}} & \multirow{2}{*}{\textbf{Method}} & \multicolumn{2}{c}{\textbf{Ellipsoidal Noise}} & \multicolumn{2}{c}{\textbf{Amplitude-Bounded}} \\ \cmidrule(lr){3-4} \cmidrule(l){5-6} & & \textbf{Max} $|x_1|$ & \textbf{Max} $|x_2|$ & \textbf{Max} $|x_1|$ & \textbf{Max} $|x_2|$ \\ \midrule \multirow{3}{*}{\begin{tabular}[c]{@{}l@{}}Case 1\\($R=10^{-4}$)\end{tabular}} & Proposed & 0.0926 & 0.1212 & 0.0929 & 0.1218 \\ & Hard-constraint Zonotope & 0.1377 & 0.1871 & 0.0849 & 0.1163 \\ & Method by \cite{mulagaleti2021data} & 0.1619 & 0.1375 & 0.1105 & 0.0925 \\ \midrule \multirow{3}{*}{\begin{tabular}[c]{@{}l@{}}Case 2\\($R=10^{-2}$)\end{tabular}} & Proposed & 0.3832 & 0.4048 & 0.3787 & 0.3985 \\ & Hard-constraint Zonotope & 0.4721 & 0.5051 & 0.2924 & 0.3135 \\ & Method by \cite{mulagaleti2021data} & 0.4447 & 0.4740 & 0.3401 & 0.3624 \\ \midrule \multirow{3}{*}{\begin{tabular}[c]{@{}l@{}}Case 3\\($R=10^{-6}$)\end{tabular}} & Proposed & 0.0885 & 0.1168 & 0.0872 & 0.1147 \\ & Hard-constraint Zonotope & 0.1290 & 0.1669 & 0.0793 & 0.1035 \\ & Method by \cite{mulagaleti2021data} & 0.1449 & 0.1169 & 0.1128 & 0.0913 \\ \bottomrule \end{tabular}} \end{table}

To further clarify the relationship with the recent data-driven RPI synthesis in \cite{wang2026data}, we add the input-state comparison under both a common feedback gain and the respective feedback gains, as shown in Fig.~\ref{fig:P1C2603_eval}.
In this experiment, $T=200$, $\gamma_w=10^{-3}$, and the input-excitation amplitude is 3.
With the respective feedback gains, the proposed certificate and the ellipsoidal certificate in \cite{wang2026data} have areas of 0.0201 and 0.2087, respectively, corresponding to a 90.38\% area reduction.
With the gain from \cite{wang2026data} fixed for both constructions, the proposed certificate has an area of 0.0267 and reduces the area by 87.21\% relative to the value 0.2087 obtained by the comparison method.

\begin{figure}[htbp]
    \centering
    \includegraphics[width=0.68\linewidth]{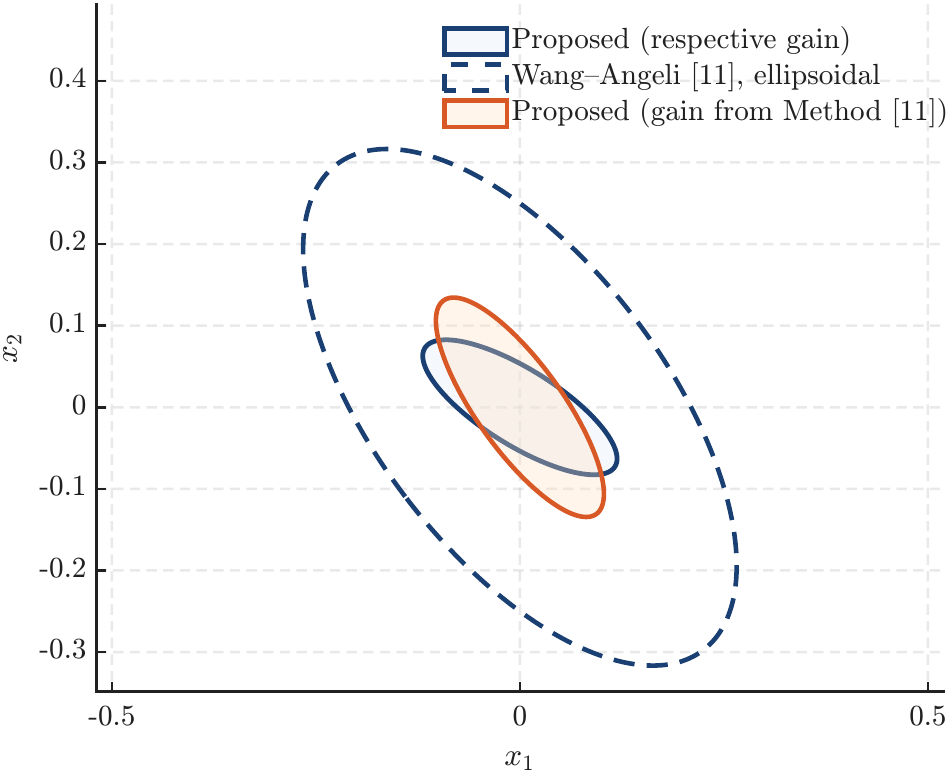}
    \caption{Additional input-state comparison against the ellipsoidal contraction-induced RPI certificate in \cite{wang2026data}.}
    \label{fig:P1C2603_eval}
\end{figure}

\subsection{Case Study II: Output-Feedback Control}
In this section, we simulate the more practical output-feedback scenario. 
To the best of our knowledge, direct data-driven synthesis of RPI sets under output measurements remains largely unexplored in the current literature. 
This is primarily because the coupling between estimation errors, measurement noise, and invariant-set variables introduces additional uncertainty channels and non-convex interactions into the joint synthesis problem.
The purpose of this simulation is to validate the effectiveness of the proposed hierarchical decoupled framework by establishing a baseline using the model-based Minkowski sum method presented in \cite{mayneRobustOutputFeedback2006}.

First, we extend the system \eqref{eq:study_1} from the previous case study by introducing an output matrix $C= \begin{bmatrix}   -0.8&0.2\\    0& 0.7  \end{bmatrix}$.
The process and measurement disturbance bounds are $Q_w^{-1}=Q_v^{-1}=\operatorname{diag}(2500,2500)$.
Using an offline trajectory of length $T=200$ with input-excitation amplitude 3, we set $Q=10^{-2}I_{n_x}$, $R=10^{-2}I_{n_u}$, $Q_e=10^{-2}I_{n_x}$, and $R_e=10^{-4}I_{n_y}$.
The known SNR bounds in Assumption~\ref{Assumption:SNR} are set to $\gamma_w=1.4877\times10^{-3}$, $\gamma_v=3.0541\times10^{-3}$, and $\epsilon^\star=1.0866$.
The grid searches select $\delta_1=0.3$, $\delta_2=0.91$, $\alpha=(0.5,0.3,0.2)$, and $\beta=(0.81,0.11,0.08)$.
Proposition~\ref{proposition:K} and Theorem~\ref{theo:L} then yield the feedback gain $K=\begin{bmatrix} -6.5525 & -2.9835 \end{bmatrix}$ and the observer gain $L=\begin{bmatrix} -0.8254 & -0.0068\\ -0.2212 & 1.0263 \end{bmatrix}$.
The resulting estimation-error and true-state RPI matrices are
\begin{equation*}
    Q_{S_e} = \begin{bmatrix}
     0.0040  &  0.0004\\
    0.0004  &  0.0051
\end{bmatrix}, \quad Q_{\tilde{S}} = \begin{bmatrix}
    0.1242 &  -0.0535\\
   -0.0535  &  0.0698\end{bmatrix}.
\end{equation*}

\begin{figure}[htbp]
    \centering
    \subfloat[RPI sets for the system in Case Study I.]{\includegraphics[width=0.8\linewidth]{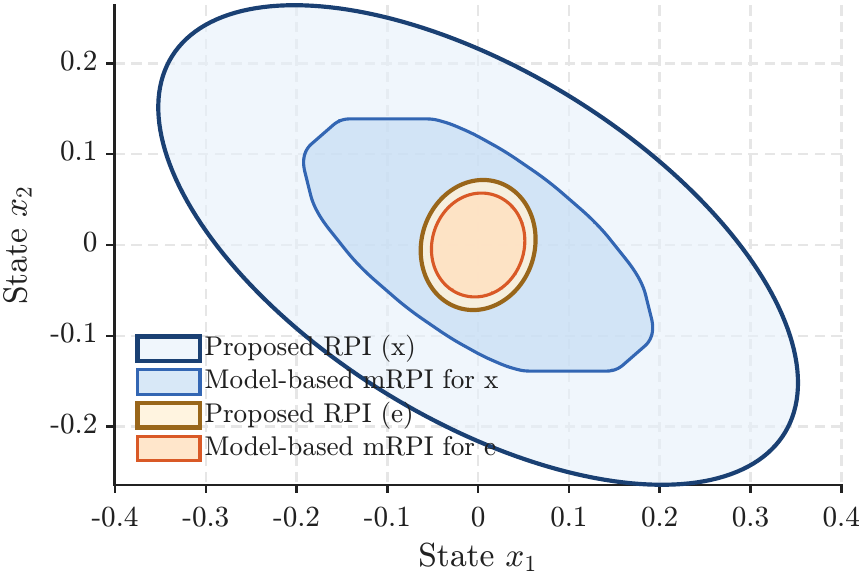}\label{fig:P2C2}}\\[0.4em]
    \subfloat[RPI sets for the double integrator system.]{\includegraphics[width=0.8\linewidth]{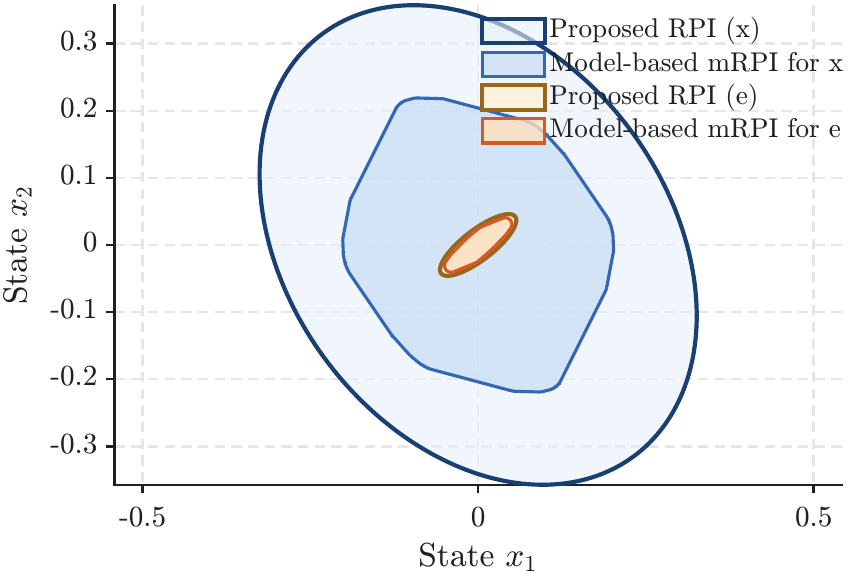}\label{fig:P2C1}}
    \caption{Synthesized true-state RPI sets ($\tilde{\mathcal{S}}_x$) and estimation-error RPI sets ($\mathcal{S}_e$) under the output-feedback architecture.}
    \label{fig:output_feedback_eval}
\end{figure}

For Fig.~\ref{fig:P2C2}, the proposed and model-based sets have areas of 0.2394 and 0.0688 for the physical state, giving a ratio of 3.4812, and areas of 0.0142 and 0.0092 for the estimation error, giving a ratio of 1.5381.

To evaluate the framework on a second system, we use the double-integrator model $A=\begin{bmatrix} 1&1\\0&1 \end{bmatrix}$, $B=\begin{bmatrix} 0.5 & 1 \end{bmatrix}^T$, and $C=\begin{bmatrix} 1 & 0 \end{bmatrix}$.

The disturbance bounds are $Q_w^{-1}=\operatorname{diag}(10^4,10^4)$ and $Q_v^{-1}=10^4$.
Using an offline trajectory of length $T=200$ with input-excitation amplitude 3, we set $Q=10^{-2}I$, $R=10^{-4}I$, $Q_e=10^{-2}I$, and $R_e=10^{-4}I$.
The collected data yield $\gamma_w=1.1288\times10^{-5}$, $\gamma_v=1.4162\times10^{-5}$, and $\epsilon^\star=0.6483$.
The grid searches select $\delta_1=0.1$, $\delta_2=0.91$, $\alpha=(0.5,0.3,0.2)$, and $\beta=(0.61,0.21,0.18)$.
The resulting gains are $K=\begin{bmatrix} -0.6988 & -0.8844 \end{bmatrix}$ and $L=\begin{bmatrix} 1.5928 & 0.6550 \end{bmatrix}^T$.
The invariant-set matrices are
\begin{equation*}
    Q_{S_e} = \begin{bmatrix}
      0.0032  &  0.0021\\
    0.0021 &   0.0021
\end{bmatrix}, \quad Q_{\tilde{S}} = \begin{bmatrix}
     0.1061 &  -0.0344\\
    -0.0344  &  0.1277
\end{bmatrix}.
\end{equation*}

For Fig.~\ref{fig:P2C1}, the proposed and model-based sets have areas of 0.3493 and 0.13 for the physical state, giving a ratio of 2.6873, and areas of 0.0049 and 0.0035 for the estimation error, giving a ratio of 1.4072.